\documentclass[12pt]{article}

\usepackage[margin=1in]{geometry}
\usepackage{amsmath, amssymb, amsthm, authblk, hyperref, CJKutf8}
\usepackage[basic]{complexity}

\newtheorem{lemma}{Lemma}
\newtheorem{theorem}{Theorem}
\newtheorem{corollary}{Corollary}

\newtheorem{definition}{Definition}

\DeclareMathOperator{\Span}{span}
\DeclareMathOperator{\tr}{tr}

\usepackage[style=trad-unsrt, citestyle=numeric-comp, giveninits=true, doi=false, url=false, eprint=false, isbn=false]{biblatex}
\begin{document}

\title{Two-dimensional local Hamiltonian problem with area laws is \QMA-complete\footnote{An earlier version of this paper appeared in \cite{self}.}}

\begin{CJK}{UTF8}{gbsn}

\author{Yichen Huang (黄溢辰)\thanks{yichuang@mit.edu} \thanks{Present address: Center for Theoretical Physics, Massachusetts Institute of Technology, Cambridge, Massachusetts 02139, USA.}}
\affil{Department of Physics, University of California, Berkeley, Berkeley, California 94720, USA}

\maketitle

\end{CJK}

\begin{abstract}

We show that the two-dimensional (2D) local Hamiltonian problem with the constraint that the ground state obeys area laws is \QMA-complete. We also prove similar results in 2D translation-invariant systems and for the 3D Heisenberg and Hubbard models with local magnetic fields. Consequently, unless $\MA=\QMA$, not all ground states of 2D local Hamiltonians with area laws have efficient classical representations that support efficient computation of local expectation values. In the future, even if area laws are proved for ground states of 2D gapped systems, the computational complexity of these systems remains unclear.

\end{abstract}


\section{Introduction}

Computing ground states of local Hamiltonians is a fundamental problem in condensed matter physics. Intuitively, this problem is likely intractable because the dimension of the Hilbert space of a quantum many-body system grows exponentially with the system size. In a pioneering work \cite{KSV02} (see Ref. \cite{AN02} for an exposition), Kitaev defined the complexity class \QMA~(Quantum Merlin-Arthur) as the quantum analog of \MA~(\MA~is the probabilistic analog of \NP) and proved that the $5$-local Hamiltonian problem is \QMA-complete. This was followed by a series of papers (see Ref. \cite{Boo14} for a survey of \QMA-complete problems):
\begin{itemize}
\item The $3$-local Hamiltonian problem is \QMA-complete \cite{KR03}.
\item The $2$-local Hamiltonian problem is \QMA-complete \cite{KKR06}.
\item The 2D local Hamiltonian problem is \QMA-complete \cite{OT08}.
\item The 1D local Hamiltonian problem is \QMA-complete \cite{AGIK07, AGIK09, HNN13}.
\item The 1D translation-invariant local Hamiltonian problem is $\QMA_\EXP$-complete \cite{GI09, GI13, BCO17}.
\item The 2D Heisenberg and Hubbard models with local magnetic fields are \QMA-complete \cite{SV09}.
\end{itemize}
Consequently, unless $\MA=\QMA$, not all ground states of local Hamiltonians have efficient classical representations that support efficient computation of local expectation values. Here, the first ``efficient'' means that the classical representation is of polynomial size, and the second means that computing local expectation values from the classical representation is in \BPP. It should be emphasized that the latter ``efficient'' is essential. Indeed, the local Hamiltonian itself, which is the sum of a polynomial number of terms, is an efficient classical representation of its ground state, but this trivial representation does not support efficient computation of local expectation values unless $\BPP=\QMA$.

Entanglement, a concept of quantum information theory, has been widely used in condensed matter and statistical physics to provide insights beyond those obtained via ``conventional'' quantities. For example, it characterizes quantum chaos \cite{DKPR16, VR17, DLL18, LCB18, Hua19NPB, LG19, HG19}, criticality \cite{HLW94, VLRK03, LRV04, CC04, RM04, CC09, HM14}, and dynamics \cite{CC05, ZPP08, BPM12, KH13, Hua17, RPv19, Hua19d, Hua20}. A generic quantum many-body state obeys a volume law: The entanglement between a subsystem (smaller than half the system size) and its complement scales as the size (volume) of the subsystem \cite{Pag93, HLW06}. Perhaps surprisingly, many physical states obey an area law: The entanglement scales as the boundary (area) of the subsystem \cite{ECP10}. Area laws are of great interest in the field of quantum Hamiltonian complexity \cite{Osb12, GHLS14, Hua15} because of not only their mathematical beauty, but also their relevance to the computational complexity of 1D quantum systems. Bounded (or even a logarithmic divergence of) Renyi entanglement entropy across every bipartite cut implies \cite{VC06, SWVC08, Hua19s, DB19, Hua19p} an efficient matrix product state (MPS) representation \cite{FNW92, Vid03, PVWC07}, which underlies the celebrated density matrix renormalization group algorithm \cite{Whi92, Whi93}. Since computing local expectation values for MPS is in \P, the 1D local Hamiltonian problem with the constraint that the ground state obeys an area law is in \NP. Moreover, an intermediate result from the proof \cite{Has07, AKLV13, Hua14} of the area law for ground states of 1D gapped systems is an essential ingredient of the (provably) polynomial-time algorithm \cite{LVV14, CF16, ALVV17, algorithm, Hua1505, Hua1510} for computing these states, concluding that the 1D gapped local Hamiltonian problem is in \P.

2D (and 3D) quantum systems can host exotic states of matter, and are much more exciting and challenging. Indeed, few rigorous positive results are known in 2D. Whether an area law can be proved for ground states of 2D gapped systems is a very important open problem in the field of quantum Hamiltonian complexity (see Ref. \cite{AAG21} for recent progress in frustration-free systems). Moreover, one might ask: Which class of 2D ground states has efficient classical representations that support efficient computation of local expectation values? If such a classical representation exists, can it be computed efficiently? Much effort has been devoted to extending methods and tools from 1D to 2D \cite{VMC08, CV09, CT17, DLD17, GD17, CCX+18, HM17}. As generalizations of MPS, projected entangled pair states \cite{VC04} and the multiscale entanglement renormalization ansatz \cite{Vid07} are two classes of tensor network states. They do not \cite{SWVC07, HHEG20} and do \cite{Vid08} support efficient computation of local expectation values, respectively. It is commonly believed that \emph{physical states with area laws have efficient tensor network representations}. This belief cannot be proved because ``physical'' is not defined. We do not attempt to define it here, but rather rely on intuition to judge what is physical. In particular, ground states of local Hamiltonians are more physical than generic quantum many-body states, and translation-invariant geometrically local Hamiltonians on a lattice are more physical than generic local Hamiltonians.

Despite the belief (in italics above), Ge and Eisert \cite{GE14} recently proved that there exist states that obey area laws for all Renyi entanglement entropies but do not have efficient classical representations. The main idea of the proof is so elegant that we sketch it here. Consider the question: How large is the set of all states with area laws? A subset $S$ can be explicitly constructed such that all states in $S$ obey area laws and that $S$ is parameterized by an exponential number of independent parameters. Since the volume of $S$ is too large, a generic state in $S$ cannot be approximately represented in polynomial space. This counting approach is powerful: It applies regardless of whether the classical representation supports efficient computation of local expectation values. Therefore, a generic state in $S$ is neither a tensor network state with polynomial bond dimension nor a (non-degenerate) eigenstate of a local Hamiltonian \cite{GE14}.

In this paper, we show that the 2D local Hamiltonian problem with the constraint that the ground state obeys area laws is \QMA-complete (Theorem \ref{cor1}). We also prove similar results in 2D translation-invariant systems (Theorem \ref{cor2}) and for the 3D Heisenberg (Theorem \ref{pro1}) and Hubbard (Theorem \ref{pro2}) models with local magnetic fields. Consequently, unless $\MA=\QMA$, not all ground states of 2D local Hamiltonians with area laws have efficient classical representations that support efficient computation of local expectation values (Corollary \ref{cor}). Our results are complementary to those of Ref. \cite{GE14}. This reference considers general quantum many-body states while we restrict ourselves to ground states of local Hamiltonians, which are more physical. Technically, the counting approach (sketched above), which yields the results of Ref. \cite{GE14}, does not work in our context. We emphasize that our results do not diminish the importance of area laws. A proof of (or a counterexample to) an area law for ground states of 2D gapped systems is, of course, a landmark achievement, which probably requires the development of powerful new techniques. However, even if such area laws are proved, the computational complexity of 2D gapped systems remains unclear.

\section{Preliminaries}

Throughout this paper, standard asymptotic notations are used extensively. Let $f,g:\mathbb R^+\to\mathbb R^+$ be two functions. One writes $f(x)=O(g(x))$ if and only if there exist positive numbers $M,x_0$ such that $f(x)\le Mg(x)$ for all $x>x_0$; $f(x)=\Omega(g(x))$ if and only if there exist positive numbers $M,x_0$ such that $f(x)\ge Mg(x)$ for all $x>x_0$; $f(x)=\Theta(g(x))$ if and only if there exist positive numbers $M_1,M_2,x_0$ such that $M_1g(x)\le f(x)\le M_2g(x)$ for all $x>x_0$. Accounting for the finite precision of numerical computing, every real number is assumed to be represented by a polynomial number of bits.

We specialize the local Hamiltonian problem to geometrically local Hamiltonians on a lattice. Let $E(H)$ be the ground-state energy of a Hamiltonian $H$.

\begin{definition} [local Hamiltonian problem on a lattice]
Consider a quantum many-body system of spins (or bosons, fermions) on a lattice. We are given a geometrically local Hamiltonian $H$ with nearest-neighbor interactions (and on-site terms) and a real number $a$ with the promise that either
\begin{itemize}
\item (Yes) $E(H)\le a$ or
\item (No) $E(H)\ge a+\delta$,
\end{itemize}
where $\delta$ is some inverse polynomial in the system size. The task is to decide which is the case.
\end{definition}

As the quantum analog of \MA, \QMA~is the class of problems that can be efficiently verified by a quantum computer.

\begin{definition} [\MA] \label{d:ma}
A decision problem is in \MA~if there is a uniform family of polynomial-size classical circuits $\{C_x\}$ (one for each input instance $x$) such that
\begin{itemize}
\item If $x$ is a yes-instance, then there exists a polynomial-length bit string $y$ such that $C_x$ accepts $y$ with probability greater than $2/3$.
\item If $x$ is a no-instance, then $C_x$ accepts any polynomial-length bit string $y$ with probability less than $1/3$.
\end{itemize}
\end{definition}

\begin{definition} [\QMA~\cite{KSV02}]
A decision problem is in \QMA~if there is a uniform family of polynomial-size quantum circuits $\{Q_x\}$ such that
\begin{itemize}
\item If $x$ is a yes-instance, then there exists a polynomial-size quantum state $|y\rangle$ such that $Q_x$ accepts $|y\rangle$ with probability greater than $2/3$.
\item If $x$ is a no-instance, then $Q_x$ accepts any polynomial-size quantum state $|y\rangle$ with probability less than $1/3$.
\end{itemize}
A problem in \QMA~is \QMA-complete if every problem in \QMA~is polynomial-time reducible to it.
\end{definition}

The entanglement entropy is the standard measure of entanglement in quantum information and condensed matter theory.

\begin{definition} [entanglement entropy]
The Renyi entanglement entropy $S_\alpha$ with index $\alpha\in(0,1)\cup(1,+\infty)$ of a bipartite pure state $\rho_{AB}$ is defined as
\begin{equation}
S_\alpha(\rho_A):=\frac{1}{1-\alpha}\log\tr(\rho_A^\alpha),
\end{equation}
where $\rho_A=\tr_B\rho_{AB}$ is the reduced density matrix. Two limits are of particular interest:
\begin{equation}
S_0(\rho_A):=\lim_{\alpha\to0^+}S_\alpha(\rho_A)
\end{equation}
is the logarithm of the Schmidt rank, and
\begin{equation}
S_1(\rho_A):=\lim_{\alpha\to1}S_\alpha(\rho_A)=-\tr(\rho_A\log\rho_A)
\end{equation}
is the von Neumann entanglement entropy.
\end{definition}

\begin{definition} [area law for $S_\alpha$]
On a lattice, a state $|\psi\rangle$ obeys an area law if
\begin{equation}
S_\alpha(\rho_A)=O(|\partial A|)
\end{equation}
for any subsystem $A$, where $\rho_A=\tr_{\bar A}|\psi\rangle\langle\psi|$ is the reduced density matrix, and $\partial A$ is the set of edges connecting subsystem $A$ and its complement $\bar A$.
\end{definition}

Since $S_\alpha$ is monotonically non-increasing with respect to $\alpha$, an area law for $S_\alpha$ implies that for $S_\beta$ if $\alpha\le\beta$. In 1D, bounded (or even a logarithmic divergence of) $S_0$ across every bipartite cut implies an efficient exact (up to the truncation of real numbers) MPS representation \cite{Vid03}; bounded (or a logarithmic divergence of) $S_\alpha$ for any $0<\alpha<1$ implies an efficient MPS approximation \cite{VC06, SWVC08}.

\section{Main results} \label{main}

In this section we prove our main result: The 2D local Hamiltonian problem with the constraint that the ground state obeys area laws is \QMA-complete. We begin with a pair of technical lemmas. Recall that $E(H)$ is the ground-state energy of $H$.

\begin{lemma} \label{l:gen}
We are given a 1D local Hamiltonian
\begin{equation} \label{eq:1dh}
H'=\sum_{i=1}^{n-1}H'_{i,i+1},
\end{equation}
where $H'_{i,i+1}$ with $\|H'_{i,i+1}\|\le1$ acts on spins $i$ and $i+1$ (nearest-neighbor interaction). Then, a local Hamiltonian $H$ on a 2D square lattice can be efficiently constructed such that
\begin{itemize}
\item The ground state $|\psi\rangle$ of $H$ obeys area laws for $S_\alpha(\alpha\ge0)$.
\item $E(H)=E(H')$.
\end{itemize}
\end{lemma}

\begin{proof}
Suppose that $H'$ acts on a chain of spins, each of which has local dimension $d=\Theta(1)$. Then, $H$ acts on a 2D square lattice of size $n\times n$, and at each lattice site there is a spin of local dimension $d$. All spins are labeled by two indices $i,j$ with $1\le i,j\le n$. The coupling between spins $(i,j)$ and $(i',j')$ is denoted by $H_{i,j}^{i',j'}$, which is nonzero only if $j=j'=1$ and $|i-i'|=1$. Define
\begin{equation}
H_{i,1}^{i+1,1}=H'_{i,i+1}
\end{equation}
for $1\le i\le n-1$. Let
\begin{equation}
H=\sum_{i=1}^{n-1}H_{i,1}^{i+1,1}+\sum_{j=2}^n\sum_{i=1}^n\big(S_{i,j}^z+d/2-1/2\big),
\end{equation}
where $S_{i,j}^z$ is the $z$ component of the spin operator of spin $(i,j)$. The ground state of $H$ is
\begin{equation}
|\psi\rangle=|\psi'\rangle\otimes\bigotimes_{j=2}^n\bigotimes_{i=1}^n|0\rangle_{i,j},
\end{equation}
where $|\psi'\rangle$ is the ground state of $H'$, and $|0\rangle_{i,j}$ is an eigenstate of $S_{i,j}^z$ with eigenvalue $1/2-d/2$. It is easy to see that $H$ has the properties stated in Lemma \ref{l:gen}.
\end{proof}

\begin{lemma} \label{thm}
We are given a 1D local Hamiltonian $H'$ (\ref{eq:1dh}). Then, a local Hamiltonian $H$ on a 2D square lattice can be efficiently constructed such that
\begin{itemize}
\item $H$ is translationally invariant if and only if $H'$ is translationally invariant.
\item The ground state $|\psi\rangle$ of $H$ obeys area laws for $S_\alpha(\alpha>0)$.
\item $|E(H)-2E(H')-a|\le O(\delta)$, where $a$ is a real number that can be computed efficiently, and $\delta=1/\poly(n)$ is an arbitrary inverse polynomial in $n$.
\end{itemize}
\end{lemma}

\begin{proof}
We construct $H$ by stacking layers of $H'$ so that $H$ is translationally invariant in the direction perpendicular to the layers. Strong interlayer coupling is introduced so that the bulk of $H$ is almost trivial. The almost trivial bulk ``dilutes'' the entanglement and leads to area laws. The boundary of $H$ is nontrivial and reproduces the physics of $H'$.

\paragraph*{Construction of $H$} Suppose that $H'$ acts on a chain of spins, each of which has local dimension $d=\Theta(1)$. Then, $H$ acts on a 2D square lattice of size $n\times n$, and at each lattice site there are two spins of local dimension $d$ (you may combine the two spins into a single spin of local dimension $d^2$ if one spin per site is preferred). All spins are labeled by three indices $i,j,k$ with $1\le i,j\le n$ and $k=1,2$. The coupling between spins $(i,j,k)$ and $(i',j',k')$ is denoted by $H_{i,j,k}^{i',j',k'}$, which is nonzero only if $|i-i'|+|j-j'|=1$ (nearest-neighbor interaction). Define the terms within each layer as
\begin{equation} \label{eq:cons}
H_{i,j,k}^{i+1,j,k}=H'_{i,i+1}
\end{equation}
for $1\le i\le n-1$, $1\le j\le n$, and $k=1,2$. Define the terms between adjacent layers as
\begin{equation}
H_{i,j,2}^{i,j+1,1}=\big(\vec S_{i,j,2}\cdot\vec S_{i,j+1,1}+d^2/4-1/4\big)\Omega(n^3/\delta)
\end{equation}
for $1\le i\le n$ and $1\le j\le n-1$, where $\vec S_{i,j,k}:=(S^x_{i,j,k},S^y_{i,j,k},S^z_{i,j,k})$ is the spin vector operator of spin $(i,j,k)$, and
\begin{equation}
\vec S_{i,j,k}\cdot\vec S_{i',j',k'}:=S^x_{i,j,k}S^x_{i',j',k'}+S^y_{i,j,k}S^y_{i',j',k'}+S^z_{i,j,k}S^z_{i',j',k'}
\end{equation}
is a physical antiferromagnetic Heisenberg interaction. Since
\begin{equation}
2\vec S_{i,j,2}\cdot\vec S_{i,j+1,1}=\big(\vec S_{i,j,2}+\vec S_{i,j+1,1}\big)^2-\vec S_{i,j,2}^2-\vec S_{i,j+1,1}^2=\big(\vec S_{i,j,2}+\vec S_{i,j+1,1}\big)^2-(d^2-1)/2,
\end{equation}
the ground state of $H_{i,j,2}^{i,j+1,1}$ is a singlet (a state with zero total spin), and $E(H_{i,j,2}^{i,j+1,1})=0$. Let $H=\sum_{j=0}^nH_j$, where on the boundary,
\begin{equation}
H_0=\sum_{i=1}^{n-1}H_{i,1,1}^{i+1,1,1},\quad H_n=\sum_{i=1}^{n-1}H_{i,n,2}^{i+1,n,2}
\end{equation}
act, respectively, on spins $(i,1,1)$ and on spins $(i,n,2)$ with $1\le i\le n$; in the bulk $1\le j\le n-1$,
\begin{equation} \label{eq:bulk}
H_j=\sum_{i=1}^{n-1}\big(H_{i,j,2}^{i+1,j,2}+H_{i,j+1,1}^{i+1,j+1,1}\big)+\sum_{i=1}^nH_{i,j,2}^{i,j+1,1}
\end{equation}
acts on spins $(i,j,2)$ and $(i,j+1,1)$ with $1\le i\le n$. Since the supports of $H_j$'s are pairwise disjoint, the ground state $|\psi\rangle=\bigotimes_{j=0}^n|\psi_j\rangle$ of $H$ is a product state in the $j$ direction, where $|\psi_j\rangle$ is the ground state of $H_j$.

\paragraph*{Translational invariance} By construction, $H$ is translationally invariant in the $j$ direction, and is translationally invariant in the $i$ direction if and only if $H'$ is translationally invariant.

\paragraph*{Entanglement entropy of $|\psi\rangle$} For notational simplicity and without loss of generality, we assume that region $A$ is rectangular and consists of spins $(i,j,k)$ with indices $i_1\le i\le i_2$, $j_1\le j\le j_2$, and $k=1,2$. (It is easy to extend the argument to a region of arbitrary shape.) Since $|\psi\rangle=\bigotimes_{j=0}^n|\psi_j\rangle$ is a product state, the entanglement entropy of $|\psi\rangle$ is the sum of that of each $|\psi_j\rangle$.
\begin{itemize}
\item The entanglement of $|\psi_j\rangle$ for $j\le j_1-2$ or $j\ge j_2+1$ is exactly zero because such $|\psi_j\rangle$ is entirely in $\bar A$.
\item The Renyi entanglement entropy of $|\psi_j\rangle$ for $j=j_1-1$ or $j=j_2$ is trivially upper bounded by $O(i_2-i_1+1)$.
\item The Renyi entanglement entropy $S_\alpha(\alpha>0)$ of $|\psi_j\rangle$ for $j_1\le j\le j_2-1$ is $O(1)$. This follows from
\end{itemize}
\begin{lemma} [area law for ground states of 1D gapped systems \cite{Has07, AKLV13, Hua14}] \label{l:arealaw}
Let $|\Psi\rangle$ be the ground state of a 1D local Hamiltonian $G=\sum_{i=1}^{n-1}G_{i,i+1}$, where $G_{i,i+1}$ with $\|G_{i,i+1}\|\le1$ acts on spins $i$ and $i+1$ (nearest-neighbor interaction). Suppose that the energy gap (the difference between the smallest and the second smallest eigenvalues) of $G$ is $\Omega(1)$. The Renyi entanglement entropy $S_\alpha(\alpha>0)$ of $|\Psi\rangle$ is $O(1)$ per cut.
\end{lemma}
Indeed, $|\psi_j\rangle$ is the ground state of $H_j$, which is a 1D local Hamiltonian by combining spins $(i,j,2)$ and $(i,j+1,1)$ into a single spin for every $1\le i\le n$. Using Weyl's inequality, after rescaling $H_j$ such that the norm of each of its terms is $O(1)$, we find that its energy gap is $\Omega(1)$. Summing up the three cases above, we obtain the upper bound $O(i_2-i_1+j_2-j_1+1)$ or 2D area laws for $S_\alpha(\alpha>0)$.

\paragraph*{Ground-state energy of $H$} Since the supports of $H_j$'s are pairwise disjoint,
\begin{equation} \label{eq:eng}
E(H)=\sum_{j=0}^nE(H_j)=2E(H')+(n-1)E(H_1),
\end{equation}
where we used $E(H_0)=E(H_n)=E(H')$ and the translational invariance in the $j$ direction in the bulk. $E(H_1)$ can be estimated using the projection lemma.
\begin{lemma} [projection lemma \cite{KKR06}] \label{proj}
Let $G_1,G_2$ be two Hamiltonians acting on a Hilbert space $\mathcal H=\mathcal H^\parallel\oplus\mathcal H^\perp$. Suppose that $G_2|_{\mathcal H^\parallel}=0$ and $G_2|_{\mathcal H^\perp}\ge J>2\|G_1\|$, where $G_2|_{\cdots}$ is the restriction of $G_2$ to a subspace. Then,
\begin{equation}
E\big(G_1|_{\mathcal H^\parallel}\big)-\frac{\|G_1\|^2}{J-2\|G_1\|}\le E(G_1+G_2)\le E\big(G_1|_{\mathcal H^\parallel}\big).
\end{equation}
\end{lemma}
Let $H_1=G_1+G_2$ with
\begin{equation}
G_1=\sum_{i=1}^{n-1}\big(H_{i,1,2}^{i+1,1,2}+H_{i,2,1}^{i+1,2,1}\big),\quad G_2=\sum_{i=1}^nH_{i,1,2}^{i,2,1}.
\end{equation}
Since the supports of $H_{i,1,2}^{i,2,1}$'s are pairwise disjoint,
\begin{equation}
E(G_2)=\sum_{i=1}^nE\big(H_{i,1,2}^{i,2,1}\big)=0,
\end{equation}
and the ground state $|\phi\rangle$ of $G_2$ is a product of singlets (unique). The ground space $\mathcal H^\parallel=\Span\{|\phi\rangle\}$ is one-dimensional, and $J=\Omega(n^3/\delta)$ is the energy gap of $G_2$. Since $\|G_1\|=O(n)$,
\begin{align}
E\big(G_1|_{\mathcal H^\parallel}\big)=\langle\phi|G_1|\phi\rangle&\implies\big|E(H_1)-\langle\phi|G_1|\phi\rangle\big|\le\frac{\|G_1\|^2}{J-2\|G_1\|}=\frac{O(n^2)}{J}=\frac{O(\delta)}{n}\nonumber\\
&\implies|E(H)-2E(H')-a|\le O(\delta),
\end{align}
where $a=(n-1)\langle\phi|G_1|\phi\rangle$ can be computed efficiently as $|\phi\rangle$ is a product of singlets. The computation time is $O(n)$ if $H'$ is not translationally invariant and $O(1)$ if $H'$ is.
\end{proof}

The construction in the proof of Lemma \ref{thm} does not lead to an area law for $S_0$, because the bulk of $H$ is almost but not completely trivial. Note that $S_0$ (the logarithm of the Schmidt rank) is not a good measure of entanglement. It is discontinuous and an infinitesimal perturbation almost always leads to maximum $S_0$.

The following is the state-of-the-art \QMA-completeness result for the 1D local Hamiltonian problem.

\begin{lemma} [\cite{AGIK07, AGIK09, HNN13}] \label{1dqma}
The local Hamiltonian problem on a chain of spin-$7/2$'s with nearest-neighbor interactions is \QMA-complete.
\end{lemma}

\begin{theorem} \label{cor1}
The 2D local Hamiltonian problem with the constraint that the ground state obeys area laws for $S_\alpha(\alpha\ge0)$ is \QMA-complete.
\end{theorem}

\begin{proof}
This follows immediately from Lemmas \ref{l:gen}, \ref{1dqma}.
\end{proof}

\begin{corollary} \label{cor}
Unless $\MA=\QMA$, not all ground states of 2D local Hamiltonians with area laws for $S_\alpha(\alpha\ge0)$ have efficient classical representations that support efficient computation of local expectation values.
\end{corollary}

\begin{proof}
This follows from the observation that such a classical representation can be taken as a proof to the verifier in the definition of \MA~(Definition \ref{d:ma}).
\end{proof}

Translational invariance introduces a slight technical complication that we have to deal with. Normally, the computational complexity is measured in terms of the input size, e.g., a decision problem is in \P~if it can be solved in polynomial time in the input size. Without translational invariance, the input size of a (not necessarily geometrically) local Hamiltonian is polynomial in the system size (as a polynomial number of terms must be specified) so that the computational complexity can be equivalently measured in terms of the system size. With translational invariance, the input size of a geometrically local Hamiltonian is the logarithm of the system size (the number of bits required to represent the system size). In this case, an exponential-time algorithm (in the input size) is efficient in the sense that its running time is polynomial in the system size.

\begin{lemma} [\cite{GI09, GI13, BCO17}] \label{tiqma}
The 1D translation-invariant local Hamiltonian problem is $\QMA_\EXP$-complete (in the input size).
\end{lemma}

For brevity, we do not define $\QMA_\EXP$-completeness here, but refer the reader to the original papers \cite{GI09, GI13} for a formal definition. Ignoring some technical subtleties, ``$\QMA_\EXP$-complete (in the input size)'' in Lemma \ref{tiqma} can be intuitively understood as ``\QMA-complete in the system size.''

\begin{theorem} \label{cor2}
The 2D translation-invariant local Hamiltonian problem with the constraint that the ground state obeys area laws for $S_\alpha(\alpha>0)$ is $\QMA_\EXP$-complete (in the input size).
\end{theorem}

\begin{proof}
This follows immediately from Lemmas \ref{thm}, \ref{tiqma}.
\end{proof}

The Hamiltonians in Theorem \ref{cor2} have open boundary conditions. It is an open problem to prove the same result for 2D translation-invariant local Hamiltonians with periodic boundary conditions.

\section{Extensions}

In this section, we extend our results to the 3D Heisenberg and Hubbard models with local magnetic fields. These models are more physical than generic local Hamiltonians on a lattice.

\begin{lemma} [\cite{SV09}] \label{pro0}
On a 2D square lattice of size $n\times n$, both the ferromagnetic and antiferromagnetic spin-$1/2$ Heisenberg models
\begin{equation} \label{2dhei}
H'=\pm\sum_{\langle i',j'\rangle}\vec{\sigma}_{i'}\cdot\vec{\sigma}_{j'}-\sum_{i'}\vec{h}'_{i'}\cdot\vec{\sigma}_{i'}
\end{equation}
with $\max_{i'}|\vec{h}'_{i'}|\le\poly(n)$ are \QMA-complete, where $\langle i',j'\rangle$ denotes nearest neighbors, and $\vec{\sigma}_{i'}:=(\sigma_{i'}^x,\sigma_{i'}^y,\sigma_{i'}^z)$ is the vector of Pauli matrices at site $i'$.
\end{lemma}

\begin{theorem} \label{pro1}
On a 3D cubic lattice of size $n\times n\times n$, both the ferromagnetic and antiferromagnetic spin-$1/2$ Heisenberg models
\begin{equation} \label{3dhei}
H=\pm\sum_{\langle i,j\rangle}\vec{\sigma}_i\cdot\vec{\sigma}_j-\sum_i\vec{h}_i\cdot\vec{\sigma}_i
\end{equation}
with $\max_i|\vec{h}_i|\le\poly(n)$ and the constraint that the ground state obeys area laws for $S_\alpha(\alpha\ge1)$ are \QMA-complete.
\end{theorem}

\begin{proof}
Given a 2D Hamiltonian (\ref{2dhei}), a 3D Hamiltonian (\ref{3dhei}) can be efficiently constructed such that
\begin{itemize}
\item The ground state $|\psi\rangle$ of $H$ obeys area laws for $S_\alpha(\alpha\ge1)$.
\item $|E(H)-E(H')-a|\le O(\delta)$, where $a$ is a real number that can be computed efficiently, and $\delta=1/\poly(n)$.
\end{itemize}
Then, Theorem \ref{pro1} follows from Lemma \ref{pro0}.

We label a site in the 2D square lattice by two indices $i'=(i'_x,i'_y)$ and that in the 3D cubic lattice by three indices $i=(i_x,i_y,i_z)$ with $1\le i'_x,i'_y,i_x,i_y,i_z\le n$. Let $p=\max_{i'}|\vec{h}'_{i'}|+n\le\poly(n)$ and $q=n^4p^2/\delta$. Define the local magnetic fields in $H$ as
\begin{equation}
\vec{h}_i=
\begin{cases}
\vec{h}'_{(i_x,i_y)}\pm(0,0,1) & i_z=1\\
(0,0,q) & 2\le i_z\le n
\end{cases}.
\end{equation}

Using the projection lemma (Lemma \ref{proj}), let $H=G_1+G_2-(n-1)n^2q$ with
\begin{align}
&G_1=\pm\sum_{\langle i,j\rangle}\vec{\sigma}_i\cdot\vec{\sigma}_j-\sum_{i_x,i_y=1}^n\vec{h}_{(i_x,i_y,1)}\cdot\vec{\sigma}_{(i_x,i_y,1)},\\
&G_2=(n-1)n^2q-\sum_{i_x,i_y=1}^n\sum_{i_z=2}^n\vec{h}_i\cdot\vec{\sigma}_i
\end{align}
so that $E(G_2)=0$. The ground space
\begin{equation}
\mathcal H^\parallel=\Span\{|\phi\rangle:\langle\phi|\sigma_i^z|\phi\rangle=1,~\forall i_z\ge2,~\forall i_x,~\forall i_y\}
\end{equation}
of $G_2$ is $2^{n^2}$-dimensional, and $J=2q$ is the smallest positive eigenvalue of $G_2$. Since $\|G_1\|=O(n^2p)$,
\begin{equation}
E\big(G_1|_{\mathcal H^\parallel}\big)=E(H')\pm(3n^3-6n^2+2n)\implies|E(H)-E(H')-a|\le\frac{\|G_1\|^2}{J-2\|G_1\|}=O(\delta)
\end{equation}
for $a=\pm(3n^3+6n^2-2n)-(n-1)n^2q$.

Let $\Pi$ be the projection onto the subspace $\mathcal H^\parallel$. Since $|\psi\rangle$ is the ground state of $G_1+G_2$,
\begin{multline}
J\|(1-\Pi)|\psi\rangle\|^2\le\langle\psi|(1-\Pi)G_2(1-\Pi)|\psi\rangle=\langle\psi|G_2|\psi\rangle\le\langle\psi|(G_1+G_2)|\psi\rangle+\|G_1\|\\
=E(G_1+G_2)+\|G_1\|\le E(G_2)+2\|G_1\|=2\|G_1\|.
\end{multline}
Let $|\psi'\rangle=\Pi|\psi\rangle/\|\Pi|\psi\rangle\|$ so that
\begin{equation}
\big|\langle\psi|\psi'\rangle\big|\ge1-O(\|G_1\|/J)=1-O(n^{-2}p^{-1}\delta).
\end{equation}
Since $|\psi'\rangle\in\mathcal H^\parallel$ trivially obeys an area law for $S_1$, $|\psi\rangle$ obeys area laws for $S_\alpha(\alpha\ge1)$ due to Lemma \ref{l:cont}($\alpha=1$) and the monotonicity of $S_\alpha$ with respect to $\alpha$.
\end{proof}

If the constraint on local magnetic fields in Theorem \ref{pro1} is relaxed to $\max_i|\vec{h}_i|\le e^{O(n^3)}$, then we can take $q=e^{O(n^3)}$ so that $|\psi\rangle$ in the above proof obeys area laws for $S_\alpha(\alpha>0)$ due to Lemma \ref{l:cont}($0<\alpha<1$).

\begin{lemma} [continuity of entanglement entropy] \label{l:cont}
Let $|\psi\rangle,|\phi\rangle$ be such that
\begin{equation}
T:=\sqrt{1-\big|\langle\psi|\phi\rangle\big|^2}\le1-1/d_A,
\end{equation}
where $d_A$ is the dimension of subsystem $A$. Let $\rho_A=\tr_{\bar A}|\psi\rangle\langle\psi|$ and $\sigma_A=\tr_{\bar A}|\phi\rangle\langle\phi|$. Then,
\begin{equation} \label{eq:cont}
|S_\alpha(\rho_A)-S_\alpha(\sigma_A)|\le
\begin{cases}
T\log(d_A-1)-T\log T-(1-T)\log(1-T) & \alpha=1 \\
\frac{1}{1-\alpha}\log\big((1-T)^\alpha+(d_A-1)^{1-\alpha}T^\alpha\big)& 0<\alpha<1
\end{cases}.
\end{equation}
\end{lemma}

\begin{proof}
Since the trace distance is non-increasing under partial trace,
\begin{equation}
T=\big\||\psi\rangle\langle\psi|-|\phi\rangle\langle\phi|\big\|_1/2\ge\|\rho_A-\sigma_A\|_1/2.
\end{equation}
Then, Lemma \ref{l:cont} follows from \cite{Nie00} the continuity of $S_\alpha$ (Theorem 1 and Eq. (A.3) of Ref. \cite{Aud07}) and the fact that both expressions on the right-hand side of (\ref{eq:cont}) are monotonically increasing for $T\in[0,1-1/d_A]$.
\end{proof}

\begin{lemma} [\cite{SV09}] \label{ll}
On a 2D square lattice of size $n\times n$, the Hubbard model at half filling 
\begin{equation}
H=-t\sum_{\langle i,j\rangle,s}c^\dag_{i,s}c_{j,s}+U\sum_ic^\dag_{i,\uparrow}c_{i,\uparrow}c^\dag_{i,\downarrow}c_{i,\downarrow}-\sum_i\vec{h}_i\cdot\vec{\sigma}_i
\end{equation}
with $0<t,U,\max_i|\vec{h}_i|\le\poly(n)$ is \QMA-complete, where $c^\dag_{i,s},c_{i,s}$ are the creation and annihilation operators of an electron with spin $s\in\{\uparrow,\downarrow\}$ at site $i$, and
\begin{equation}
\vec{\sigma}_i:=\sum_{s,s'}(\sigma_{ss'}^x,\sigma_{ss'}^y,\sigma_{ss'}^z)c^\dag_{i,s}c_{i,s'}
\end{equation}
is a vector of operators with $\sigma_{ss'}^{\cdots}$ the matrix elements of the Pauli matrices.
\end{lemma}

\begin{theorem} \label{pro2}
On a 3D cubic lattice of size $n\times n\times n$, the Hubbard model at half filling
\begin{equation}
H=-t\sum_{\langle i,j\rangle,s}c^\dag_{i,s}c_{j,s}+U\sum_ic^\dag_{i,\uparrow}c_{i,\uparrow}c^\dag_{i,\downarrow}c_{i,\downarrow}-\sum_i\vec{h}_i\cdot\vec{\sigma}_i
\end{equation}
with $0<t,U,\max_i|\vec{h}_i|\le\poly(n)$ and the constraint that the ground state obeys area laws for $S_\alpha(\alpha\ge1)$ is \QMA-complete.
\end{theorem}

\begin{proof}
Using Lemma \ref{ll} instead of Lemma \ref{pro0}, Theorem \ref{pro2} can be proved in the same way as Theorem \ref{pro1}.
\end{proof}

\section*{Declaration of competing interest}

The author declares that he has no known competing financial interests or personal relationships that could have appeared to influence the work reported in this paper.

\section*{Acknowledgments}

This work was supported by the DARPA OLE program.

\printbibliography

\end{document}